\documentclass[pra,aps,reprint,a4paper,superscriptaddress,floatfix,longbibliography]{revtex4-2}
\usepackage[english]{babel}
\usepackage{times}
\usepackage{graphicx}
\usepackage{amsfonts}
\usepackage{amsmath}
\usepackage{amssymb}
\usepackage{bbm}
\usepackage{amsthm}
\usepackage{color}
\usepackage{xcolor}
\usepackage[colorlinks=true,linkcolor=blue,citecolor=magenta,urlcolor=blue]{hyperref}
\usepackage{physics}
\usepackage{mathtools}
\usepackage{tabularx}

\newtheorem{lemma}{Lemma}
\newtheorem{observation}{Observation}

\newcommand{\exs}[1]{\ensuremath{\langle{#1}\rangle}}
\newcommand{\CZ}[2]{\text{CZ}_{(#1,#2)}}
\newcommand{\B}[2]{\mathcal{B}_{#1}^{\text{#2}}}
\newcommand{\Q}[2]{Q_{#1}^{\text{#2}}}
\newcommand{\C}[2]{C_{#1}^{\text{#2}}}
\newcommand{\D}[2]{D_{#1}^{\text{#2}}}
\newcommand{\Es}[1]{\ensuremath{\mathbb{E}[#1]}}                
\newcommand{\Prob}[1]{\ensuremath{\mathbb{P}\left(#1\right)}}      
\newcommand{\Probs}[1]{\ensuremath{\mathbb{P}(#1)}}                
\newcommand{\Id}{\mathbbm{1}}
\newcolumntype{C}{>{\centering\arraybackslash}X}

\begin{document}


\title{Generating multipartite nonlocality to benchmark quantum computers}


\author{Jan Lennart B\"onsel}
\email{jan.boensel@uni-siegen.de}
\author{Otfried Gühne}
\email{otfried.guehne@uni-siegen.de}
\affiliation{Naturwissenschaftlich-Technische Fakult\"at, Universit\"at Siegen, Walter-Flex-Stra{\ss}e 3, 57068 Siegen, Germany}

\author{Ad\'an Cabello}
\email{adan@us.es}
\affiliation{Departamento de F\'{\i}sica Aplicada II, Universidad de Sevilla, 41012 Sevilla, Spain}
\affiliation{Instituto Carlos~I de F\'{\i}sica Te\'orica y Computacional, Universidad de Sevilla, 41012 Sevilla, Spain}


\date{\today}


\begin{abstract}
We show that quantum computers can be used for 
producing large $n$-partite nonlocality, thereby
providing a method to benchmark them.
The main challenges to overcome are as follows: (i)~The interaction topology might 
not allow arbitrary two-qubit gates. (ii)~Noise limits the Bell violation. (iii)~The 
number of combinations of local measurements grows exponentially with $n$. 
To overcome (i), we point out that graph states that are compatible with the 
two-qubit connectivity of the computer can be efficiently prepared. 
To mitigate (ii), we note that for specific graph 
states, there are $n$-partite Bell inequalities whose resistance to 
white noise increases exponentially with $n$. 
To address 
(iii) for any $n$ and any connectivity, we introduce an estimator that relies on random sampling. 
As a result, we propose a method for producing $n$-partite Bell nonlocality with unprecedented large $n$. 
This allows one, in return, to benchmark nonclassical correlations regardless of the number of qubits or the connectivity.
We test our approach by using a simulation for a noisy IBM quantum computer, which
predicts $n$-partite Bell nonlocality for at least $n=24$ qubits.
\end{abstract}


\maketitle


\section{Introduction}


The term ``$n$-partite nonlocality'' refers to correlations between $n$ parties 
that cannot be explained by any local realistic model \cite{Brunner2014, Bancal2013}.
It can be detected by a violation of a multipartite Bell inequality, which shows that at least
one of the assumptions of a local realistic model is false \cite{Bell1964, Vieira2024}.
The experimental test of $n$-partite nonlocality is, however, challenging.
The first and main reason is that it is very difficult to have physical systems with $n$ parts 
that can be prepared in a genuinely $n$-partite entangled quantum state \cite{Mermin1990} and 
on which specific local measurements can be performed on each of the $n$ individual parts. 
From this perspective, quantum computers offer a unique chance to realistically go to a large $n$ 
and test quantum theory. 
Quantum computers have dozens, hundreds, or even thousands of qubits which can, in principle, 
be prepared in arbitrary quantum states and then measured individually.

We propose a method to produce and certify $n$-qubit Bell nonlocality with 
unprecedented large $n$.
Quantum mechanics predicts nonlocality and that the violation increases exponentially with $n$.
The main motivation is thus to experimentally test this prediction for large $n$, 
i.e., in the ``macroscopic'' limit.
Specifically, in our case, the aim is observing nonlocality produced by
``a superposition of macroscopically distinct states'' \cite{Mermin1990}.
In this respect, this work is in the line of recent results showing that quantum computers can produce 
correlations that are impossible in other platforms \cite{Chen2022} or used for many-body 
simulation of fermionic systems \cite{Jafferis2022}.

The second motivation is to use the observed Bell violation as a benchmark to compare 
different quantum computers.
Since, as far as we know, $n$-partite nonlocality is a phenomenon specific to quantum 
theory, one can think of using it to quantify the ``quantumness'' of the device
that has produced it. 
This can be achieved by the fraction $D_{n}=Q_{n}/C_{n}$ of
the maximal quantum value $Q_n$ and the classical bound $C_{n}$
\cite{Cabello2008, Werner2001, Guehne2005, Mermin1990}.
For our use-case, we stress that $D_{n}$ is related to the resistance of the violation to 
depolarization noise.
Moreover, $D_{n}$ is associated to the detection efficiency that is required to 
classically simulate the quantum nonlocality \cite{Cabello2008}.

A variety of benchmarks have been proposed to test the quality of quantum computers, 
i.e., the quantum volume or the cross-entropy benchmark \cite{Eisert2020, Frank2024, Arute2019}.
Still, no universally accepted standard has been established. Current approaches do 
not fulfill the ideal requirements to be independent of the noise model and the hardware:
to not be tied to one algorithm but still being predictive and scalable in practice \cite{Frank2024}. 
The phenomenon of $n$-partite nonlocality is promising in this 
regard, as a Bell violation has an interpretation independent of the hardware and 
the specific noise. Observing a violation of a Bell inequality requires specific 
quantum states and measurements. Consequently, Bell inequalities can be used to
certify both measurements and states \cite{Supic2020}, which makes them attractive for 
benchmarking. 
In contrast, entanglement witnesses certify a quantum state given
some well-characterized measurements, which requires additional assumptions. 
In addition, a Bell test can be carried out, in principle, for all pure entangled states 
\cite{Popescu1992}. In this sense, using Bell inequalities for benchmarking does not 
rely on a specific algorithm to prepare a certain quantum state. The fraction $D_{n}$ 
grows with system size and we show that this facilitates the scaling of the Bell test 
to many qubits. Finally, an observed Bell violation can be used to lower bound the 
fidelity, which allows one to predict the quality of other computations.

Experimentally, violations of $n$-partite Bell inequalities have been observed in a 
variety of physical systems that are promising for quantum computing
\cite{Lanyon2014,Zhang2015,Pelisson2016,Alsina2016,
Swain2019,Gonzalez2020,Baumer2021,Amouzou2022,Yang2022,DeBoutray2021,Singh2022}. 
So far, however, mostly systems with a relatively small number of parties, $n$, have been considered. 
In an ion trap, violations of up to $n=14$ have been verified \cite{Lanyon2014}, whereas
nonlocality has also been shown for $n=6$ with photons \cite{Zhang2015}.
$n=3$ nonlocality has been verified on a nuclear magnetic resonance (NMR) quantum simulator \cite{Singh2022}.
Finally, nonlocality has also been investigated in atomic ensembles \cite{Engelsen2017} and in optical
lattices \cite{Pelisson2016}.
The largest number of parties has been achieved with superconducting qubits with up to $n=57$ 
\cite{Yang2022}.
To limit the experimental resources, however, this reference considered Bell inequalities 
that do not show an exponential violation in $n$.

The question remains why exponentially increasing nonlocality has not been 
verified with quantum computers before. 
There are, fundamentally, three challenges:

(i)~Typically, quantum computers can apply two-qubit gates only on some specific pairs of qubits. 
Hereafter, we will refer to the map that specifies these pairs as the connectivity of the 
quantum computer. 
A consequence of this limitation is that not all connectivities allow us to equally well 
prepare an $n$-qubit Greenberger-Horne-Zeilinger (GHZ) \cite{Greenberger1989} state, 
which was the default option in \cite{Alsina2016,Lanyon2014}.

(ii)~Quantum computers with larger $n$ are typically more affected by noise and decoherence. 
Therefore, the larger $n$ is, the harder it becomes to prepare the target state and to observe 
a violation of a Bell inequality. 

(iii)~The number of different combinations of local measurements (contexts) needed to test a 
Bell inequality increases with $n$.
For example, in the case that there are two measurements per qubit, the number of contexts 
scales exponentially in $n$ and, typically, so does the number of terms needed to test the Bell inequality.
Therefore, measuring all contexts becomes infeasible for large $n$.

In this work, we show how to overcome or alleviate each of these three challenges. First, 
we will show that there is a natural solution to problem (i).
For a given connectivity, we can focus on
the graph states that are compatible with the connectivity graph. 
Graph states are a specific set of pure entangled states \cite{Hein2006} and can 
be prepared by applying controlled-Z (CZ) gates on adjacent qubits in the graph.
Thus, if we assume that CZ gates can be performed between connected qubits, graph states
that correspond to subgraphs of the two-qubit connectivity can be readily prepared.

For $n$-qubit graph states, there are some general methods to obtain $n$-partite Bell inequalities 
\cite{Guehne2005, Guehne2008, Cabello2008,Scarani2005,Geza2006}.
However, it is, in general, a hard task to find the optimal one (in the sense of resistance to noise 
of the violation); the optimal $n$-partite Bell inequalities in terms of the stabilizers have been 
identified for some graph states \cite{Cabello2008}. 
In particular, the optimal $n$-partite Bell inequalities associated to the GHZ and 
linear cluster (LC) state are known \cite{Cabello2008, Guehne2008}.
These states correspond to the extreme cases of connectivity: 
On the one side, the GHZ state is easy to prepare when all qubits can be coupled to each other.
On the other side, we consider the LC state that can be conveniently prepared on
a quantum computer with minimal connectivity, i.e., the connectivity graph is a line.
Cluster states have, in addition, the advantage to be more resistant to decoherence \cite{Duer2004}.
Quantum theory predicts that for the GHZ and LC states, the ratio $D_{n}$ can be made 
arbitrarily large by increasing the number of qubits 
\cite{Mermin1990,Ardehali1992,Guehne2008}.
This means that in theory, the resistance to noise of multipartite nonlocality {\em grows 
exponentially} with the number of particles.
This helps to overcome (ii).

In addition, the main aim of this work is to introduce a general method to address problem (iii).
For this purpose, we discuss how the expectation value of a Bell operator can be estimated
from the measurements of only a few terms.
The terms are chosen at random and thus the method falls into line with previous randomized
measurement approaches, e.g., direct fidelity estimation \cite{Flammia2011, Cao2023}
or few-copy entanglement verification \cite{Saggio2018}.
As this method is not restricted to a specific Bell inequality, it can be applied to the 
Bell operator that is most appropriate for the experimental set-up taking into account
the feasible interactions.

We note, however, that quantum computers usually are not suited for a loophole-free Bell test.
For example, ions are typically in the same trap and superconducting qubits on the same chip.
It is thus not possible to rule out communication. 
However, in principle, the interactions can be tuned to minimize the crosstalk between the qubits. 
Hereafter, we will refer to this assumption as the no-crosstalk assumption and 
we will make it on the belief that quantum computers are the only way to investigate 
$n$-partite nonlocality with large $n$. 

To explain the proposed method, we will introduce graph states that can be readily prepared
on a given connectivity. 
Accordingly, we discuss examples of Bell inequalities associated to graph states that show
an exponential scaling of the fraction $D_{n}$ in Sec.~\ref{Sec_Bell_ineqs}.
Afterwards, we will explain in Sec.~\ref{Sec_methods} the method to measure the Bell inequalities.
As we propose to evaluate the Bell inequalities by random sampling in Sec.~\ref{Sec_random_sampling},
we first formulate the Bell test as a hypothesis test in Sec.~\ref{Sec_hypothesis_test}.
After we discuss the sample complexity in Sec.~\ref{Sec_state_samples}, we will
apply the method exemplarily to the Bell inequalities of the GHZ and LC states
in Sec.~\ref{Sec_Analysis_LC_GHZ}.
These Bell inequalities cover the extreme cases of connectivities in quantum computers.
In doing so, we consider actual architectures of current quantum computers.
Finally, we include a simulation for the IBM Eagle quantum processor in Sec.~\ref{Sec_simulation_IBM}.


\section{Bell inequalities with an exponential nonlocality}\label{Sec_Bell_ineqs}


We start this section by discussing graph states, which is an important class of quantum states
in quantum information theory \cite{Hein2006, Guehne2009}.
For our aim, they are specifically useful as the graph states that are compatible with the
connectivity graph of an $n$-qubit quantum computer can be readily prepared.
Suppose $G$ is a graph of $n$ vertices. 
For each vertex $i$ in $G$, we define a stabilizing operator $g_{i}$ by
\begin{equation}\label{Eq_stabilizers}
    g_{i}\coloneqq X_{i}\bigotimes_{j\in\mathcal{N}(i)}Z_{j},
\end{equation}
where $\mathcal{N}(i)$ is the neighborhood of vertex $i$, i.e., all vertices that are connected
to $i$.
In the above definition, $X_{i}$, $Y_{i}$ and $Z_{i}$ denote the Pauli matrices acting on qubit $i$.
The graph state $\ket{G}$ that is associated with $G$ is the common eigenstate of all stabilizing operators
with eigenvalue $+1$, i.e.,
\begin{equation}
    g_{i}\ket{G}=\ket{G},\qquad \text{for }i=1,\ldots,n.
\end{equation}
The graph state $\ket{G}$ has the explicit expression \cite{Hein2006}
\begin{equation}\label{Eq_graph_state_explicit}
    \ket{G}=\prod_{(i,j)\in E}\CZ{i}{j}\ket{+}^{\otimes n},
\end{equation}
where $E$ is the set of edges of the graph and $\CZ{i}{j}$ is the controlled-Z gate acting
on qubits $i$ and $j$.
This motivates the assumption of the specific universal gate set.
Quantum computers that natively implement the $\text{CZ}$ gate can readily prepare the
graph state in Eq.~\eqref{Eq_graph_state_explicit} in case the graph $G$ is a subgraph
of its connectivity graph.
\begin{figure}[t!]
    \includegraphics[width=0.9\linewidth]{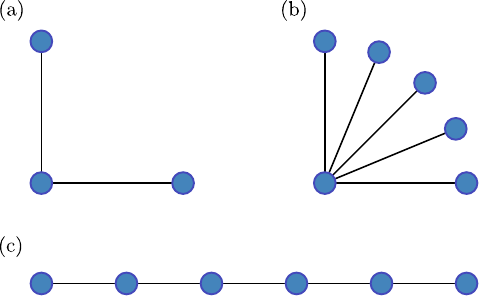}
    \caption{Three exemplary graph states. (a) For three qubits, the GHZ state and the LC state
        coincide. The star graph in (b) corresponds to the GHZ state of six qubits, whereas
        the line graph in (c) is associated to the LC state of six qubits.}
    \label{Fig_graph_states}
\end{figure}

Different graphs may lead to graph states that are connected by a local unitary (LU) 
transformation, which does not change the entanglement properties.
A special class of LU transformations is local Clifford operations \cite{Hein2006}, which
can be described by a graphical rule called local complementation.
A local complementation in vertex $i$ transforms a graph $G=(V,E)$ into 
a graph $G'=(V,E')$ by inverting the neighborhood $\mathcal{N}(i)$ of vertex $i$.
For two vertices $j,k\in\mathcal{N}(i)$, if $(j,k)\in E$ then $(j,k)\notin E'$, and 
vice versa.
The set of vertices $V$ is unchanged.

In addition, graph states also have another benefit.
For all graph states, there is a known associated Bell inequality that is maximally
violated by the graph state \cite{Guehne2005}.
Specifically, for tree graphs, i.e., connected graphs that do not contain a cycle,
the violation increases exponentially in the number of qubits, $n$ \cite{Guehne2005},
i.e., the ratio $\D{n}{}$ increases exponentially in $n$.
This facilitates the observation of a Bell violation.
Suppose, in an experiment, that the noisy graph state $\rho=\alpha\ketbra{G}+(1-\alpha)\Id/2^{n}$ is
prepared.
The expectation value of the Bell operator is $\exs{\B{n}{}}_{\rho}=\alpha\Q{n}{}$.
A violation is thus observed for
\begin{equation}\label{Eq_violation_p_lbound}
    \alpha>\frac{\C{n}{}}{\Q{n}{}}=D_{n}^{-1},
\end{equation}
which decreases exponentially with the number of qubits $n$.

We conclude this section by discussing two specific Bell inequalities for the graph
states in Fig.~\ref{Fig_graph_states}.
For the $n$-qubit GHZ state, Mermin's inequality constitutes a known Bell inequality 
\cite{Mermin1990, Guehne2008}.
Mermin's inequality is up to local rotations equivalent to the Bell operator
\begin{equation}\label{Eq_Bell_operator_GHZ}
    \B{n}{GHZ}=g_{1}\prod_{i=2}^{n}(\Id+g_{i}),
\end{equation}
where $g_{i}$ are the stabilizing operators defined in Eq.~\eqref{Eq_stabilizers}.
The quantum bound $\Q{n}{GHZ}=2^{n-1}$ is achieved for the $n$-qubit GHZ state.
The classical bound and thus the maximal ratio $\D{}{}$ for Mermin's inequality are
\begin{equation}
\begin{split}
    \C{n}{GHZ}=\begin{cases}
    2^{\frac{n-1}{2}},\\
    2^{\frac{n}{2}},
    \end{cases},~
    \D{n}{GHZ}=\begin{cases}
    2^{\frac{n-1}{2}}, & n~\text{odd},\\
    2^{\frac{n-2}{2}}, & n~\text{even}.
    \end{cases}
\end{split}
\end{equation}

For the LC state, in turn, we consider the Bell inequalities in
Ref.~\cite{Guehne2008} that are defined in case the number of qubits is a multiple of three.
The Bell operator is 
\begin{equation}\label{Eq_Bell_operator_LC}
    \B{n}{LC}=\prod_{\substack{i=1}}^{n/3}(\Id+g_{3i-2})
        g_{3i-1}(\Id+g_{3i}),
\end{equation}
which takes the maximal value for the $n$-qubit LC state.
As all terms in Eq.~\eqref{Eq_Bell_operator_LC} are stabilizing operators of
the LC state, we have $\Q{n}{LC}=4^{n/3}$.
The classical bound, in turn, is
\begin{equation}
    \C{n}{LC}=2^{n/3}\qquad\text{and thus}\qquad\D{n}{LC}=2^{n/3}.
\end{equation}


\section{Methods}\label{Sec_methods}


The Bell operators in Sec.~\ref{Sec_Bell_ineqs} rely on a number of contexts 
that scales exponentially in the number of qubits $n$, i.e., there
are, in principle, an exponential number of terms to measure.
The number of observables thus quickly becomes infeasible.
In this section, we introduce a general method to evaluate a Bell inequality
by sampling random terms.
The terms of the Bell inequality are picked according to a uniform probability distribution.
This approach stands in line with previous schemes that use randomization to reduce
the number of measurements, e.g., direct fidelity estimation \cite{Flammia2011, Cao2023}
or few-copy entanglement detection \cite{Saggio2018}.
Finally, there exist different methods to assess the statistical strength
of Bell tests \cite{Peres2000, VanDam2005, Zhang2010, Elkouss2016}.
We gauge the significance of a violation with the help of the $p$ value.
For this purpose, we start by formulating a Bell test as a hypothesis test.


\subsection{Hypothesis test}\label{Sec_hypothesis_test}


The task is to evaluate a general Bell inequality with classical bound $C$, i.e.,
to check the inequality
\begin{equation}
    \exs{\B{}{}}\leq C.
\end{equation}
The expectation value on the right-hand side, however, cannot be inferred exactly in 
an experiment.
Rather, the expectation value has to be estimated from multiple experimental repetitions.
For this purpose, it is useful to consider an unbiased estimator $\exs{\hat{\B{}{}}}$, which
we denote by a hat.
An estimator $\exs{\hat{\B{}{}}}$ is a function of the experimental data and it
is unbiased in case it reproduces the actual value in expectation, i.e., 
$\Es{\exs{\hat{\B{}{}}}}=\exs{\B{}{}}$.

Due to the finite statistics, however, the estimate $\exs{\hat{\B{}{}}}$ 
fluctuates. 
There is thus a nonzero probability to observe a violation of the Bell inequality,
although the actual state does not violate it.
To quantify the probability that an observed violation is only due to statistical
fluctuations, we formulate the Bell test as a hypothesis test.
The hypotheses are as follows:
\begin{itemize}
    \item[(a)] Null hypothesis \textbf{$\mathbf{H_{0}}$:} The measurement outcomes can be described by a local
        hidden variable (LHV) model.
    \item[(b)] Alternative hypothesis \textbf{$\mathbf{H_{1}}$:} The Bell inequality is violated.
\end{itemize}

To gauge whether the observed data is in contradiction with the null hypothesis $H_{0}$,
we look at the $p$ value.
The $p$ value is defined as the probability to observe an estimate at least as large as some value 
$\theta$ in case $H_{0}$ is true, i.e.,   
\begin{equation}
    p=\Prob{\exs{\hat{\B{}{}}}>\theta~\big|~H_{0}}.
\end{equation}
But the $p$ value is hard to calculate since the probability depends on the 
probability distribution of the estimator $\exs{\hat{\B{}{}}}$, which is unknown.
We can, however, upper bound the $p$ value.
A local theory can, at most, reach a value of $\exs{\B{}{}}=\Es{\exs{\hat{\B{}{}}}}=C$.
Thus, in case $H_{0}$ is true, the estimator $\exs{\hat{\B{}{}}}$ has to exceed its mean by at least
$t=\theta-C$ if a violation $\exs{\hat{\B{}{}}}=\theta>C$ is observed.
This is sketched in Fig.~\ref{Fig_p_value_ubound}.
We consequently obtain the upper bound
\begin{equation}\label{Eq_p_value_ubound}
    p\leq\Probs{\exs{\hat{\B{}{}}}-\exs{\B{}{}}\geq t}.
\end{equation}
In the following, we use Hoeffding's inequality \cite{Hoeffding1963} to upper bound 
the right-hand side of Eq.~\eqref{Eq_p_value_ubound}.
Hoeffding's inequality is a large deviation bound that typically involves the 
number of repetitions.
This allows us to connect the number of repetitions to the $p$ value.
Finally, we say that an observed result has a confidence level of $\gamma=1-p$.


\begin{figure}[t!]
    \includegraphics[width=0.9\linewidth]{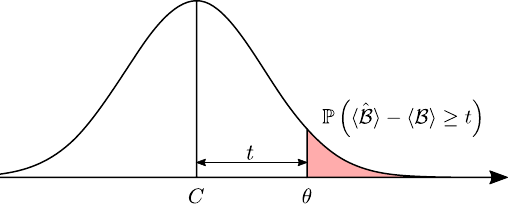}
    \caption{Upper bound of the $p$ value. To observe a value of $\theta>C$ even
        though the state obeys the classical bound, the estimator $\exs{\hat{\mathcal{B}}}$
        has to exceed its mean by at least $t$.
    \label{Fig_p_value_ubound}}
\end{figure}


\subsection{Random sampling}\label{Sec_random_sampling}


So far, we have noted that there are known multipartite Bell inequalities
that adapt to the architecture of the quantum computer and show a growing resistance 
to white noise.
This allows one to overcome problems (i) and (ii).
We have also seen, however, that the number of terms grows exponentially, 
which makes it experimentally infeasible to measure all terms.
In this section, we describe a scheme to overcome this problem and which allows one to
evaluate the Bell inequalities.

A general Bell operator $\mathcal{B}$ can be written as a sum of observables,
\begin{equation}
    \B{}{}=\sum_{j=1}^{M}B_{j}.
\end{equation}
We note that for the Bell inequalities in Sec.~\ref{Sec_Bell_ineqs}, the number of terms equals
the quantum bound, i.e., $M=Q_{n}$.
We propose to estimate the expectation value of the Bell operator
by measuring the expectation value of $L$ randomly chosen terms.
To analyze how many terms $L$ have to be sampled, we first assume that the
expectation values can be inferred directly, i.e., we consider the limit
of infinite measurements.
With this simplification, the estimator reads
\begin{equation}\label{Eq_estimator_inf}
    \exs{\hat{\mathcal{B}}}_{\infty}=\frac{M}{L}\sum_{l=1}^{L}\exs{B_{J_{l}}}.
\end{equation}
We note that in the above expression, $J_{1},\ldots ,J_{L}$ are independent random variables
with possible outcomes in the range $\{1,\ldots,M\}$.
We assume that all outcomes are equally likely, i.e., 
$\Probs{J_{l}=j}=\frac{1}{M}$ for all $j\in\{1,\ldots ,M\}$.
In Appendix~\ref{App_estimator_inf}, we show that the estimator is unbiased, i.e.,
$\Es{\exs{\hat{\mathcal{B}}}_{\infty}}=\exs{\mathcal{B}}$.

To assess the significance of an observed violation, we consider the $p$ value.
We upper bound the $p$ value with the help of Eq.~\eqref{Eq_p_value_ubound}.
The probability on the right-hand side of Eq.~\eqref{Eq_p_value_ubound}
can be bounded with the help of concentration inequalities.
Here, we consider Hoeffding's inequality, which yields
\begin{equation}\label{Eq_p_value_est_inf}
    p\leq\Probs{\exs{\hat{\mathcal{B}}}_{\infty}-\exs{\mathcal{B}}\geq t}
    \leq\exp\left(-\frac{t^{2}}{2M^{2}}L\right).
\end{equation}

The details of Hoeffding's inequality are presented in Appendix~\ref{App_Hoeffding_inf}.
This result can be used to derive a lower bound on the necessary $L$.
To reach a confidence level of $\gamma=1-p$, 
\begin{equation}\label{Eq_necessary_L}
    L\geq\left\lceil-\frac{2M^{2}}{t^{2}}\ln(1-\gamma)\right\rceil
\end{equation}
random terms of the Bell operator have to be sampled.


\subsection{Number of measurement repetitions}\label{Sec_state_samples}


We now take into account that the expectation values cannot be inferred directly.
Rather, the measurement of each chosen term $B_{j}$ has to be repeated multiple times.
In this section, we are interested in how many repetitions are necessary.
For this purpose, we adjust the estimator in Eq.~\eqref{Eq_estimator_inf} to include 
the measurement repetitions.
We assume that every term is measured $K$ times. 
This yields the estimator
\begin{equation}\label{Eq_estimator_rep}
    \exs{\hat{\mathcal{B}}}=\frac{M}{KL}\sum_{l=1}^{L}\sum_{k=1}^{K}b_{J_{l}}^{(k)}.
\end{equation}
In the above estimator, $b_{j}^{(k)}$ denotes the measurement outcome of the term $B_{j}$
in the $k$th repetition.
Thus, we have $b_{j}^{(k)}\in\{\pm 1\}$.
As in Eq.~\eqref{Eq_estimator_inf}, $J_{1},\ldots ,J_{L}$ denote independent random 
variables that when uniformly distributed, take values in the range $\{1,\ldots,M\}$.
We also show in Appendix~\ref{App_estimator_rep} that the estimator in Eq.~\eqref{Eq_estimator_rep}
is unbiased, i.e., $\Es{\exs{\hat{\mathcal{B}}}}=\exs{\mathcal{B}}$.
Finally, we can again use Eq.~\eqref{Eq_p_value_ubound} and Hoeffding's inequality 
to bound the $p$ value,
\begin{equation}\label{Eq_p_value_Ubound_rep}
    p\leq\Probs{\exs{\hat{\mathcal{B}}}-\exs{\mathcal{B}}\geq t}
    \leq\exp(-\frac{t^{2}}{2M^{2}}KL).
\end{equation}
However, we have already derived a lower bound for $L$.
We thus obtain a lower bound on $K$ from Eq.~\eqref{Eq_p_value_Ubound_rep}:
\begin{equation}
    K\geq\frac{1}{L}\left\lceil-\frac{2M^{2}}{t^{2}}\ln(1-\gamma)\right\rceil.
\end{equation}
Since $L\geq\lceil-\frac{2M^{2}}{t^{2}}\ln(1-\gamma)\rceil$, we observe
that $K\geq 1$.
As a result, we can conclude that it is sufficient to measure each term
only once.
\begin{figure}[t]
    \includegraphics[width=\linewidth]{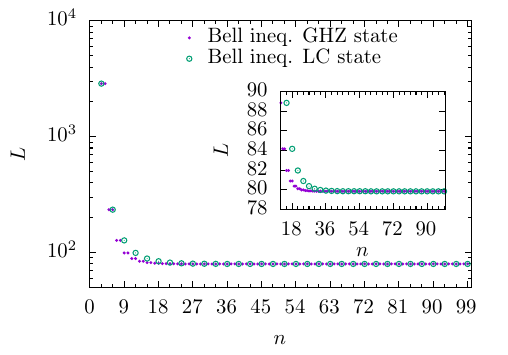}
    \caption{Necessary number of random observables, $L$, which are measured $K=1$ times each, 
    such that an observed violation of at least $\exs{\hat{\B{}{}}}=\alpha Q_{n}$ for 
    $\alpha=0.6$ has a confidence of $\gamma=5\sigma$.
    $L$ is plotted as a function of the number of qubits, $n$.
    \label{Fig_Example_Necessary_L_GHZ_LC_alpha0_6}}
\end{figure}

As an example, we consider the case that a violation of at least 
$\exs{\hat{\B{}{}}}=\alpha Q_{n}$ for $\alpha=0.6$ is observed.
This value is motivated by Ref.~\cite{Cao2023}, where they prepared a $51$-qubit LC state
with fidelity of approximately $0.6$.
The necessary number of measurement settings, $L$, is shown in Fig.\,\ref{Fig_Example_Necessary_L_GHZ_LC_alpha0_6}.
We note that for small $n$, $L$ exceeds the number of contexts of the Bell operator.
This, however, is not a contradiction as we do not 
exclude that a term is sampled multiple times.
As an example, we can make the following observation:
\begin{observation}
    In case a violation $\exs{\hat{\B{}{}}}=0.6\times Q_{n}$ of the Bell inequality 
    associated to the GHZ state or the LC state for $n=51$ qubits has been
    observed by sampling $L=80$ random terms, the result has a confidence level
    of $\gamma=5\sigma$.
\end{observation}


\section{Analysis of the Bell inequalities for the GHZ and LC state}\label{Sec_Analysis_LC_GHZ}


In this section, we are going to apply the method to the Bell inequalities for graph states.
In particular, we will look at the Bell inequalities for the GHZ and the LC states.
As described in the introduction, these Bell inequalities are promising to detect a large 
violation. 
The discussion in the previous section has shown that this is necessary to verify the 
violation from a few measurement settings with high significance.
We start with the question of how the LC and GHZ states can be prepared on a quantum computer with
a given two-qubit connectivity.
Afterwards, we consider the effect of noise on the preparation with a simple depolarization
noise model, which gives insight into the sample complexity of the method.


\subsection{Connectivities of current quantum computers}


We start by having a look at the connectivity graphs of different
quantum computers and the Bell inequalities that can be evaluated on the 
different architectures.
In Fig.~\ref{Fig_architectures}, we show the connectivity graphs of a few current
quantum computers.
The first connectivity in Fig.~\ref{Fig_architectures}(a) is the star graph of
five qubits, i.e., one central qubit connected to four other qubits.
This layout is used, e.g., in the Starmon-5 quantum processor \cite{QuTech2020}.
Figure~\ref{Fig_architectures}(b) shows the connectivity graph that is used
by IBM's Falcon processor \cite{IBMQ}.
The ion trap quantum computer in Ref.~\cite{Friis2018} has 20-qubits that can all be 
coupled.
The corresponding connectivity graph is shown in Fig.~\ref{Fig_architectures}(c).
We also include the connectivity graphs of Google's Sycamore processor in 
Fig.~\ref{Fig_architectures}(d) \cite{Arute2019} that has $53$ qubits
and IBM's Eagle processor \cite{IBMQ} that has $127$ qubits in Fig.~\ref{Fig_architectures}(e).
\begin{figure}[!b]
    \includegraphics[width=\linewidth]{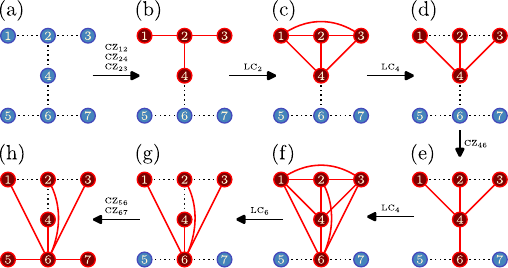}
    \caption{Preparation scheme for the GHZ state on the $7$-qubit
        interaction topology that is, for example, used by IBM's Falcon processor.
        The dotted lines indicate the physical CZ gates, whereas the graph state
        is drawn in red.
    \label{Fig_graph_state_prep}}
\end{figure}

\begin{figure*}[!t]
    \includegraphics[width=\linewidth]{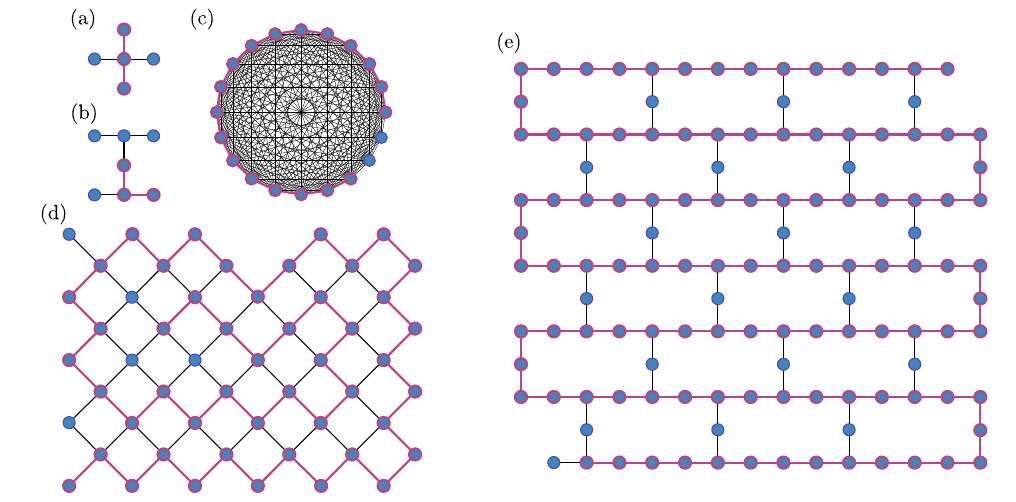}
    \caption{Connectivities of different quantum computers. The connectivity graph in
    (a) is used, for example, by the the Starmon-5 quantum processor \cite{QuTech2020},
    whereas IBM's Falcon processor \cite{IBMQ} is based on layout (b).
    (c) The connectivity of the ion trap quantum computer in Ref.~\cite{Friis2018}.
    (d) The connectivity graph of Google's Sycamore processor
    \cite{Arute2019}; (e) IBM's Eagle processor \texttt{ibm\_brisbane} \cite{IBMQ}.  
    \label{Fig_architectures}}
\end{figure*}

As the optimal Bell inequality that is associated to the connectivity graph is,
in general, hard to determine, we will focus on the Bell inequalities
for the GHZ and the LC states. 
In the following, we are interested in the largest
GHZ and LC states that can be prepared on the different layouts in Fig.~\ref{Fig_architectures}.
Ideally, the preparation only requires the set of basis gates that are directly implemented by the quantum computer.
\begin{observation}
    On a quantum computer of $n$ qubits, it is always possible to prepare a GHZ state 
    of all $n$ qubits with a circuit of $\mathcal{O}(n)$ depth.
\end{observation}
\begin{proof}
    The connectivity graph of quantum computers is usually connected.
    Thus, by the following steps, a GHZ state of all $n$ qubits can be prepared.
    The steps are illustrated in Fig.~\ref{Fig_graph_state_prep} for 
    the architecture in Fig.~\ref{Fig_architectures}(b).
    \begin{enumerate}
        \item[(1)] Prepare a star graph with center at the qubit with the largest connectivity
            [Fig.~\ref{Fig_graph_state_prep}(b)].
        \item[(2)] By performing two local complementations, the center of the star graph can be
            shifted to any node of the graph.
            Thus, the center can be moved to a vertex with still uncoupled neighbors
            [Figs.~\ref{Fig_graph_state_prep}(c) and \ref{Fig_graph_state_prep}(d)].
        \item[(3)] By applying a CZ gate between the center node and the uncoupled neighbors
            the adjacent qubits can be added to the GHZ state
            [Fig.~\ref{Fig_graph_state_prep}(e)].
        \item[(4)] Step (2) and (3) can be repeated until all qubits are coupled.
    \end{enumerate}
    This procedure requires at least $n-1$ consecutive CZ gates.
    In the worst case, there are two local complementations needed between
    all CZ gates. Combined with the initial Hadamard gates, the circuit depth
    is $3n+1$.
\end{proof}
We point out that the GHZ state can also be prepared in logarithmic step complexity 
\cite{Cruz2019, Yu2023}, depending on the connectivity.

For the LC state, in contrast, we make the following observation.
\begin{observation}
    Assume that the connectivity graph of a quantum computer is connected. 
    Then, a linear cluster state containing all $n$ qubits can be prepared with
    a circuit depth of $\mathcal{O}(n)$.
    In practice, however, it is often beneficial to prepare the LC state that
    is associated to the longest simple path in the connectivity graph \cite{Cao2023}.
    The corresponding circuit has a constant depth of three, independent of 
    the number of qubits.
\end{observation}
\begin{proof}
    We show in Appendix~\ref{App_prep_LC_state} that it is always possible to 
    prepare a LC state that contains all qubits in linear circuit depth.
    The longest simple path, in contrast, can be generated by first preparing
    all qubits in the $\ket{+}$ state, i.e., by applying Hadamard gates.
    Afterwards, every second CZ gate can be performed in parallel.
    The circuit depth is thus three, independent of the length of the simple
    path. 
    The circuit is shown for the $6$-qubit LC state in 
    Fig.~\ref{Fig_qcircuits_LC_GHZ_qiskit}(a).
\end{proof}

As we only know the Bell inequality for the LC state with a number of qubits that is a multiple
of three, we search for the longest path of length divisible by three.
The longest paths that fulfill this restriction are drawn in pink in Fig.~\ref{Fig_architectures}.
The connectivity graphs in Figs.~\ref{Fig_architectures}(a) and 
\ref{Fig_architectures}(b) allow one to prepare a $3$-qubit LC state.
The largest simple path on the $20$-qubit full-connectivity graph in 
Fig.~\ref{Fig_architectures}(c) with length divisible by three has length $18$.
The quantum computer in Fig.~\ref{Fig_architectures}(c) thus allows one to check
the Bell inequality for the $18$-qubit LC state.
To find the longest path is an NP-complete problem \cite{Schrijver2002}.
We can thus not verify if the marked paths for the layouts in 
Figs.~\ref{Fig_architectures}(d) and \ref{Fig_architectures}(e)
are indeed the longest paths.
In Fig.~\ref{Fig_architectures}(d), we have identified a $48$-qubit path as the
longest simple path.
The layout in Fig.~\ref{Fig_architectures}(e), in turn, allows the preparation of
the $108$-qubit LC state.

The use of the LC state corresponding to the longest single path
implies that we are not using {\em all} the qubits of the quantum computer.
Of course, if one focuses on a single graph state that only covers {\em part} of the qubits, then one obtains only {\em partial} information about the computer. 
However, in principle, one can cover all the qubits by LC and GHZ states and, in this way, obtain more information. The quality of the benchmark depends on the effort and detail one wants to invest in it. 
But, even with a moderate effort, our method allows making simple statements such as quantum computer $A$ manages to produce a value $\exs{\hat{\mathcal{B}}}$ of the Bell parameter with $N$ qubits in a LC state, which are valuable to compare $A$ to other computers.


\subsection{Noise}


In the following section, we will discuss the effect of noise.
The goal is to roughly estimate the violation that can be realistically observed on
quantum computers and how the violation scales with the number of qubits, $n$.
Commonly, the errors of quantum computers are specified in terms of the error rates
for single-qubit gates, two-qubit gates, and readout.
The average error rates for IBM's Eagle processor \texttt{ibm\_brisbane} and Google's
Sycamore processor are shown in Table~\ref{Tab_error_rates}.
The error rate for the single-qubit gates is typically about an order smaller than
the other error rates.
The readout error on the other side can be mitigated by classical postprocessing
\cite{Maciejewski2020,Cai2023}.
To consider the effect of noise on the Bell violation, we will 
consider a simple depolarization noise model.


\begin{table}[!b]
    \centering
    \begin{tabular}{lccc}
         \hline
         \hline\\[-6pt]
         & Single-qubit gate & Two-qubit gate & Readout\\
         IBM Eagle & $4.322\times 10^{-4}$ & $1.019\times 10^{-2}$ & $2.434\times 10^{-2}$ \\[2pt]
         \hline\\[-6pt]
         Google Sycamore &  &  & \\
         \quad isolated & $1.5\times 10^{-3}$ & $3.6\times 10^{-3}$ & $3.1\times 10^{-2}$ \\
         \quad simultaneous & $1.6\times 10^{-3}$ & $6.2\times 10^{-3}$ & $3.8\times 10^{-2}$\\[2pt]
         \hline
         \hline
    \end{tabular}
    \caption{Average error rates of IBM's Eagle processor \texttt{ibm\_brisbane} 
        \cite{IBMQ}
        and Google's Sycamore processor \cite{Arute2019}. Google gives the error rates for 
        the cases that the gates are performed isolated or simultaneously on all qubits.
        The error rates are averaged over all gates or the readout 
        error of all qubits.}
    \label{Tab_error_rates}
\end{table}


In the depolarization noise model, we assume that an error results, on average, in 
a maximally mixed state, i.e., the circuit for the graph state $\ket{G}$ prepares
the mixture
\begin{equation}
    \rho=\alpha\ketbra{G}+(1-\alpha)\frac{\Id}{2^{n}}.
\end{equation}
The probability that no error in the preparation
occurs is $\alpha=(1-p_{1})^{N_{1}}(1-p_{2})^{N_{2}}(1-p_{r})^{n}$,
where $p_{1}$ and $p_{2}$ are the average error rates for single- and two-qubit gates
and $p_{r}$ is the average readout error.
$N_{1}$ denotes the number of single-qubit gates. 
We note that idling qubits can be seen as an identity gate acting on the qubits.
Identity gates are single-qubit gates and the error can also be described by the
average error in the single-qubit gates \cite{IBMQ}.
$N_{2}$, in turn, is the number of two-qubit gates and $n$ the number of qubits that
are prepared.
The observed violation is thus
\begin{equation}
    \expval{\mathcal{B}}=\alpha\times Q_{n}.
\end{equation}


\begin{figure}[t]
    \centering
    \includegraphics{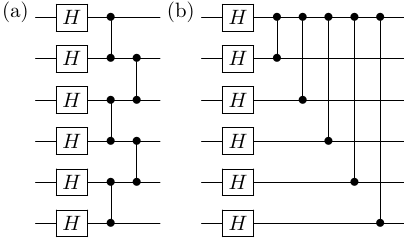}
    \caption{Quantum circuits to prepare (a) the LC state and (b) the GHZ state for
        $n=6$ qubits. The circuits make use of the Hadamard gate denoted by $H$ and CZ
        operations that are represented by the connected dots.
    \label{Fig_qcircuits_LC_GHZ_qiskit}}
\end{figure}


In Fig.~\ref{Fig_qcircuits_LC_GHZ_qiskit}, we show the circuits to prepare the LC and
GHZ states in the case of $n=6$ qubits.
For the LC state, the CZ gates can be performed simultaneously and thus the circuit
exhibits a constant depth of three, independent of the number of qubits.
The number of single-qubit gates is $N_{1}=n+2$, which includes two identity operations
for the idling qubits.
There are, in total, $n-1$ CZ gates, such that $N_{2}=n-1$.
For the GHZ state, however, the preparation with CZ gates requires the gates to be consecutively applied.
The circuit depth thus grows with the number of qubits.
In each step, $n-2$ of the qubits are idling.
Therefore, there are, in total, $(n-1)(n-2)$ identity operations, such that the preparation
requires, in total, $N_{1}=n+(n-1)(n-2)$ single qubit gates.
The number of CZ gates is equal, i.e., $N_{2}=n-1$.


\begin{table}[b]
    \centering
    \begin{tabularx}{\linewidth}{lCC}
         \hline
         \hline\\[-6pt]
         & Violation $\exs{\mathcal{B}}/Q$ & $L$ for $\gamma=5\sigma$ \\[2pt]
         \hline\\[-6pt]
         Google Sycamore &  & \\
         \quad LC state ($n=48$) & $0.1073$ & $1696$ \\
         \quad GHZ state ($n=53$) & $0.0982$ & $2024$ \\[2pt]
         IBM \texttt{ibm\_brisbane} &  & \\
         \quad LC state ($n=108$) & $0.0223$ & $39402$ \\
         \quad GHZ state ($n=127$) & $0.0139$ & $100501$ \\[2pt]
         \hline
         \hline
    \end{tabularx}
    \caption{Results for the depolarization noise model. For the noise data of Google's Sycamore
        \cite{Arute2019} and IBM's Eagle processor \cite{IBMQ}, we show the violations
        in terms of the quantum bound $Q$. The last column shows the number of random terms,
        $L$, that have to be sampled to verify the violation with a confidence of $5\sigma$.}
    \label{Tab_depolarisation_noise_model}
\end{table}


The results in Table~\ref{Tab_depolarisation_noise_model} show
that the simple noise model predicts a violation of approximately $0.1\times Q$
for Google's Sycamore processor and around $0.02\times Q$ for \texttt{ibm\_brisbane}.
We note, however, that we considered more qubits on the IBM machine.
Moreover, the violation of the GHZ state is smaller compared to the LC state.
This indicates that the GHZ state is more affected by noise.
A violation can accordingly be verified with 
$L\sim 2000$ measurement settings on Google's Sycamore processor.
On the IBM Eagle processor, in turn, $L\sim 40000$ settings are needed to verify
nonlocality in the LC state and $L\sim 100000$ for the GHZ state.
Real quantum computers, moreover, do not implement all gates natively.
IBM's Eagle processor, for example, does not support the CZ gate.
Rather, it has to be composed of the available gates.
The circuit in practice thus contains more gates and possibly exhibits a larger depth.
For these reasons, the noise model overestimates the violation.
However, the noise model is still useful to assess the scaling of the sample complexity.


\subsection{Sample complexity}


\begin{figure}[b]
    \includegraphics[width=\linewidth]{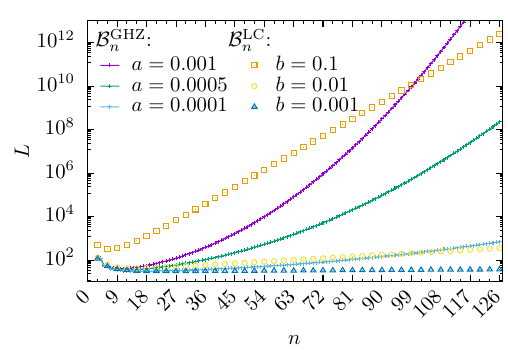}
    \caption{Necessary number of random observables $L$, which are measured $K=1$ times each, 
    such that an observed violation of at least $\exs{\hat{\B{}{}}}=\alpha Q_{n}$ 
    with $\alpha_{\text{GHZ}}=\exp(-an^{2})$ or $\alpha_{\text{LC}}=\exp(-bn)$
    has a confidence of $\gamma=5\sigma$.
    $L$ is plotted as a function of the number of qubits, $n$.
    \label{Fig_Necessary_L_GHZ_LC}}
\end{figure}


In the previous section, we have seen that noise causes the 
Bell violation to decline exponentially with the number of qubits, i.e., in leading order, we have
\begin{subequations}\label{Eq_scaling}
\begin{align}
    \alpha_{\text{GHZ}}&=\exp(-an^{2}),\\
    \alpha_{\text{LC}}&=\exp(-bn)
\end{align}
\end{subequations}
with $a,b>0$.
As we have seen in Eq.~\eqref{Eq_violation_p_lbound}, a violation of the Bell
inequality can still be observed for large $n$ if $\alpha$ vanishes slower than the fraction
$D_{n}^{-1}$, i.e.,
$\alpha>D_{n}^{-1}\overset{n\rightarrow\infty}{\rightarrow}0$.
We note that for the LC state, this is the case for $b<\ln(2)/3$.
In this parameter regime, the relative violation $\exs{\hat{\mathcal{B}}}/C$ of the Bell 
inequality for the LC state is still increasing with $n$.
For the GHZ state, however, the scaling in Eq.~\eqref{Eq_scaling} shows that 
the effect of noise increases faster than the quantum bound.
A violation of the Bell inequality for the GHZ state can thus only be observed for
$n<\ln(2)/(4a)+\sqrt{(\ln(2)/(4a))^{2}-\ln(2)/a}$.
In case a violation is observed, it is $t=\alpha Q_{n}-C_{n}$.
The number of random terms that is necessary to ensure a confidence level $\gamma$
is given by Eq.~\eqref{Eq_necessary_L} and takes the form
\begin{equation}\label{Eq_L_GHZ_LC}
    L\geq\left\lceil-\frac{2}{(\alpha -D_{n}^{-1})^{2}}\ln(1-\gamma)\right\rceil.
\end{equation}
We note that for the Bell inequalities in Sec.~\ref{Sec_Bell_ineqs}, the 
number of terms equals the quantum bound, i.e., $M=Q_{n}$.
In case $\alpha$ also decreases exponentially, the number of random settings, $L$,
increases exponentially. This is shown in Fig.~\ref{Fig_Necessary_L_GHZ_LC}.

For both the GHZ state as well as the LC state, we have plotted $L$ as a function 
of the number of qubits, $n$.
The scaling for the GHZ state is due to the number of single-qubit gates. 
We therefore choose the parameter $a$ of the same order as the error in the single-qubit gates,
i.e., $a=0.001, 0.0005$, and $0.0001$.
For the LC state, in contrast, all errors contribute to the leading term and we choose
$b=0.1,0.01$, and $0.001$.
Figure~\ref{Fig_Necessary_L_GHZ_LC} shows that for $a\leq 0.0001$ or $b\leq0.01$, 
the number of required measurements is feasible on current quantum computers.

\begin{figure}[!t]
    \includegraphics[width=\linewidth]{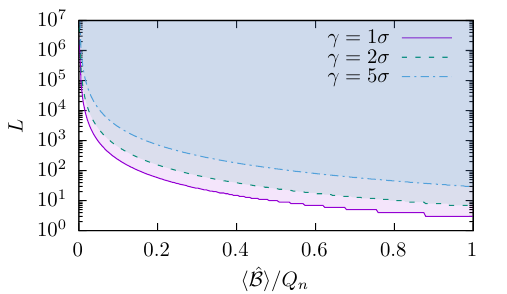}
    \caption{Confidence levels $\gamma$ for an observed violation $\exs{\hat{\B{}{}}}$
        with $L$ random terms ($K=1$) in case the classical bound $C_{n}$ is negligible
        compared to the observed violation $\exs{\hat{\mathcal{B}}}$.
    \label{Fig_confidence_level_large_n}}
\end{figure}

In addition, Eq.~\eqref{Eq_L_GHZ_LC} shows that with increasing $L$, a
decreasing violation $\exs{\hat{\B{}{}}}=\alpha Q_{n}$ with $\alpha\sim\mathcal{O}(L^{-1/2})$
can be verified.

Finally, if the classical bound $C_{n}$ is negligible compared to the observed violation, it does not have an effect on $L$.
We thus show the contours of the confidence levels $\gamma=1\sigma,2\sigma$, and $5\sigma$
in this limit for a given observation $\exs{\hat{\B{}{}}}$ with 
$L$ random terms in Fig.~\ref{Fig_confidence_level_large_n}.


\section{Simulation for an IBM quantum computer}\label{Sec_simulation_IBM}


\begin{figure*}[!t]
    \centering
    \includegraphics[width=\linewidth]{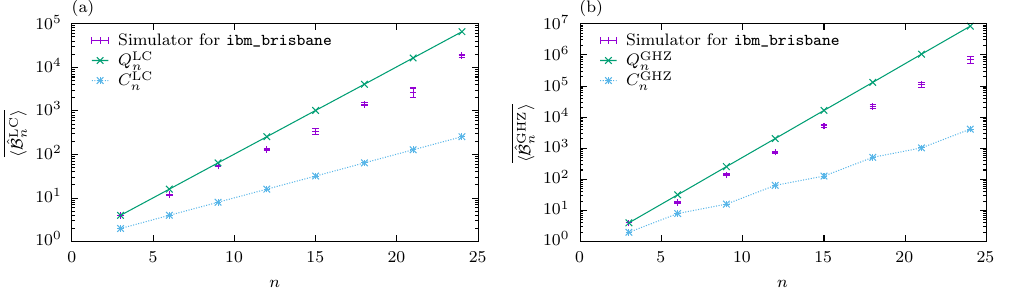}
    \caption{Simulation for the IBM Eagle quantum processor. The simulation uses the error rates 
    of the real device \texttt{ibm\_brisbane} \cite{IBMQ} and the violation is estimated by 
    measuring $L=800$ random terms of the Bell inequality $K=1$ times each. 
    In (a), the average expectation value of the Bell inequality $\mathcal{B}^{\text{LC}}$ is shown 
    for the LC state, whereas (b) shows the violation of $\mathcal{B}^{\text{GHZ}}$ for the GHZ state. 
    In both cases, the expectation values are averaged over $10$ repetitions and
    the error bars show the standard deviation.
    \label{Fig_IBM_noisy_simulation_LC_GHZ}}
\end{figure*}


Finally, we simulate the Bell inequalities of the LC and the GHZ states for
the IBM Eagle quantum processor.
For this purpose, we use the Qiskit AerSimulator \cite{Qiskit} with the noise data of the
quantum computer \texttt{ibm\_brisbane} available at \cite{IBMQ}.
In Fig.~\ref{Fig_qcircuits_LC_GHZ_qiskit}, we show the ideal circuits
to prepare the LC and the GHZ states for $n=6$ qubits.
The advantage of the LC state is that it can be prepared by a circuit of constant depth
of three, whereas the step complexity for the GHZ state increases with the number of qubits,
$n$.
We note, however, that IBM's Eagle processor does not implement the Hadamard and CZ gates 
natively.
Rather, the gates have to be composed in terms of the available gate set.
In practice, the circuits thus involve more gates and exhibit a larger depth.

After the preparation, $L$ random terms of the corresponding Bell inequality are measured.
The measurement of each random term is not repeated, i.e., $K=1$.
Figure~\ref{Fig_IBM_noisy_simulation_LC_GHZ} shows the average expectation values of the Bell inequalities 
for the LC and GHZ states of up to $n=24$ qubits.
We have chosen $L=800$ random terms and the average is taken over $10$ repetitions.
Moreover, the number of qubits is a multiple of three as only in this
case is a good Bell inequality for the LC state known.
Figure~\ref{Fig_IBM_noisy_simulation_LC_GHZ} shows that for both states, the simulation
predicts a Bell violation that increases exponentially with $n$.
The LC state, however, shows a slightly higher relative violation compared to the GHZ state.
This can be seen in Fig.~\ref{Fig_IBM_noisy_simulation_LC_GHZ_extrapolation}(a).
The plot in Fig.~\ref{Fig_IBM_noisy_simulation_LC_GHZ_extrapolation}(a) displays the observed
expectation value as a fraction of the quantum bound, i.e.,
$\exs{\hat{\mathcal{B}}}/Q$.
In agreement with the scaling that is predicted by the depolarization noise model, 
we fit the logarithmic data to a linear function for the LC state and a quadratic 
function for the GHZ state.
The fits yield
\begin{subequations}\label{Eq_violation_exp_fit}
\begin{align}
    \frac{\exs{\hat{\mathcal{B}}_{n}^{\text{LC}}}}{Q_{n}^{\text{LC}}}&=\exp(-0.078n+0.248),\\
    \frac{\exs{\hat{\mathcal{B}}_{n}^{\text{GHZ}}}}{Q_{n}^{\text{GHZ}}}&=\exp(-0.001n^{2}-0.078n+0.154).
\end{align}
\end{subequations}
This affirms that the relative violation of the GHZ state decreases faster with $n$
compared to the LC state.
We attribute this to the larger circuit depth that is required for the GHZ state.
The preparation of the GHZ state is thus more affected by noise.
The smaller relative violation is also the reason for the larger $p$ values 
for the GHZ state, which we present in Fig.~\ref{Fig_IBM_noisy_simulation_LC_GHZ_pvalue}.
Except for the cases of $n=3,15$, the $p$ values of the LC state are smaller, which implies
a higher significance of the observed violation.
Overall, the $p$ values show a large variation.
We note that the $p$ values depend on the observed violation.
Apparently, for some numbers of qubits, $n$, the evaluation of the Bell inequalities 
is more affected by noise.
We attribute this to the mapping of the actual architecture.
For different sizes of the state, the optimal choice of qubits might be differently 
affected by noise.
This also explains the variation of the error bars in Fig.\,
\ref{Fig_IBM_noisy_simulation_LC_GHZ}.
In case the evaluation is more affected by noise, not only is the observed violation smaller
but it also appears reasonable that the value fluctuates more.
In addition, we note that the error bars in 
Fig.~\ref{Fig_IBM_noisy_simulation_LC_GHZ_extrapolation}(a) do not match
the fitted exponential. 
This is because the plotted error bars show the standard deviation of the repeated
simulations and do not cover all uncertainties. 
For example, the error bars do not include the uncertainties due to the different
mappings of the quantum circuits to the actual architecture.


\begin{figure}[!b]
    \centering
    \includegraphics{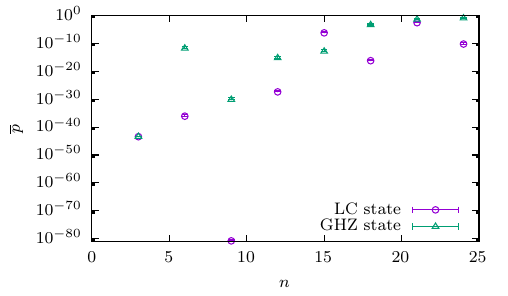}
    \caption{Average $p$ value of the results for the LC and GHZ states in 
        Fig.~\ref{Fig_IBM_noisy_simulation_LC_GHZ}. The average is taken over the $p$ values
        of the $10$ repetitions. The error bars denote the standard deviation. 
        We plot only the top error bars as we are interested in the uncertainty to
        larger $p$ values.
    \label{Fig_IBM_noisy_simulation_LC_GHZ_pvalue}}
\end{figure}


\begin{figure*}[!t]
    \centering
    \includegraphics[width=\linewidth]{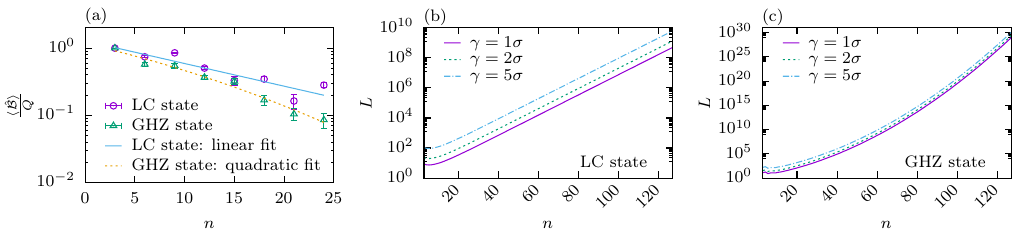}
    \caption{(a) Observed violation as a fraction of the quantum bound, i.e.,
        $\exs{\hat{\mathcal{B}}}/Q$. The logarithmic data are fitted
        by a linear function for the LC state and a quadratic function for the GHZ state.
        From the fitted function, we estimate the number of random terms, $L$,
        that have to be sampled to reach a confidence of $\gamma=1\sigma,2\sigma,5\sigma$. (b) The
        results for the LC state; (c) the predictions for the GHZ state.
    \label{Fig_IBM_noisy_simulation_LC_GHZ_extrapolation}}
\end{figure*}


Finally, we use Eq.~\eqref{Eq_violation_exp_fit} to extrapolate the violation to
larger $n$.
We note that the circuits have to be adapted to the architecture of the quantum
computer.
For this task, we use the automatic transpilation provided in Qiskit, 
which involves an optimization such that the qubits are chosen that are least 
affected by noise.
The extrapolation overestimates the violation for large $n$, since in this case such a choice is no longer possible.
Hoeffding's inequality in Eq.~\eqref{Eq_p_value_Ubound_rep} yields the following for $K=1$ and 
a target value for the confidence $\gamma$:
\begin{equation}
    L(n,\gamma)\geq \left\lceil-\frac{2}{(t^{2}/M^{2})(n)}\ln(1-\gamma)\right\rceil.
\end{equation}
The number of necessary sampled terms, $L$, is shown in 
Fig.~\ref{Fig_IBM_noisy_simulation_LC_GHZ_extrapolation}(b) for
the LC state and in Fig.~\ref{Fig_IBM_noisy_simulation_LC_GHZ_extrapolation}(c) for the
GHZ state.

As an example, we discuss $L$ for the case of the $n=108$-qubit LC state, i.e.,
the largest simple path for the architecture. 
For better comparison, we also calculate $L$ for the $n=108$-qubit GHZ state.
To reach a target confidence of
$\gamma=5\sigma$, the values are
\begin{subequations}
\begin{align}
    L_{\text{LC}}(n=108,\gamma=5\sigma)&=330997173,\\
    L_{\text{GHZ}}(n=108,\gamma=5\sigma)&=1.825\times 10^{23}.
\end{align}
\end{subequations}
The nonlocality of the LC state can be verified by a large number of measurement settings.
For the GHZ state, in turn, the extrapolation also predicts a violation.
To verify this violation, however, requires an
infeasible number of measurements.
This shows again that the GHZ state is more affected by noise than the LC state and that the LC state is thus more promising to detect large multipartite nonlocality.
The above numbers are still smaller than $M_{\text{LC}}\approx 10^{21}$
and $M_{\text{GHZ}}\approx 10^{32}$, but much larger than the values predicted in
Sec.~\ref{Sec_Analysis_LC_GHZ}.
We attribute this to the fact that the generic depolarization noise model does not cover all the noise in a specific quantum computer.
Recall, however, that our main goal is to detect large-scale quantum nonlocality and our method can still be used to assess the confidence of an observed
Bell violation.
Of course, the verification with a certain confidence may require more measurements for specific computers. 
In the worst case (e.g., when $\alpha$ drops and the $p$ value is not sufficient), one just needs to make additional measurements.


\section{Conclusion}\label{Sec_conclusion}


We have demonstrated a method to test $n$-partite Bell nonlocality on quantum computers.
Quantum computers often have restricted two-qubit connectivity.
We have thus pointed out that graph states are a natural choice of nonlocal states
that can be readily prepared if the graph is a subgraph of the connectivity graph.
Moreover, for certain graph states, good Bell inequalities are known.
These Bell inequalities allow for an exponential violation of the classical bound, 
but, in turn, also typically require an exponential number of measurements.
On the one hand, the exponential violation makes them increasingly robust to noise.
On the other hand, it is impossible to measure all terms in an experiment.
We have solved this problem by proposing a method in the manner of randomized measurements, e.g.,
direct fidelity estimation \cite{Flammia2011,Cao2023} or few-copy entanglement 
detection \cite{Saggio2018}. 
By sampling the terms of the Bell operator at random, the number of measurements can 
be drastically reduced.
The violation can, however, still be verified with high significance.
We have gauged the significance of a result with Hoeffding's inequality. 
It thus dependents on the violation that is observed.
To assess the usefulness of the method on real devices, we have first used a simple 
depolarization noise model to estimate realistic violations.
Finally, we have simulated the method for the IBM Eagle quantum processor.
As expected with increasing accuracy of the noise model, the predicted violation shrinks.
However, also, the simulator of the IBM processor predicts the number of terms that have to be sampled
to be much smaller than the total number of terms in the Bell inequalities.

Our method will hence be useful to verify Bell violations in quantum systems of many
qubits.
This includes current quantum computers in the NISQ regime, e.g., the quantum computers
accessible at IBM Quantum \cite{IBMQ}.
In addition, the observed Bell violation can be used to benchmark and compare different 
quantum computers.
The Bell violation can be interpreted as a measure for the nonclassical correlations 
that can be produced.
The preparation of the associated state depends on the connectivity of the quantum computer.
We thus can benchmark the nonclassical correlations for states that require 
different levels of two-qubit connectivity.
Finally, we stress that our method is not restricted to qubits and can be readily applied
to Bell inequalities with higher local dimension.

Furthermore, the method could also be refined. 
For example, as the Bell inequalities only include stabilizers of the graph state,
all observables commute.
It might thus be feasible to find a (possibly very complicated) positive operator valued measure (POVM) to simultaneously measure 
all of the terms.

Moreover, it could be interesting to analyze the Bell inequalities with other 
statistical methods.
Instead of the $p$ value, one might look at the Kullback-Leibler divergence that has been used to 
assess the statistical strength of Bell inequalities for few parties \cite{VanDam2005}.

The relation to other benchmarks, e.g., the quantum volume \cite{Moll2018,Baldwin2022} 
or the layer fidelity \cite{Mckay2023}, is also yet to be explored.
In particular, the layer fidelity can be measured by benchmarking a linear string of qubits of the quantum computer.

\textit{Note added}. Recently, similar ideas have been
discussed in \cite{Wang2024}.


\section*{Acknowledgements}


The authors would like to thank
Lina Vandr\'e,
H. Chau Nguyen,
Mariami Gachechiladze,
Konrad Szyma\'nski,
Ties Ohst,
Kiara Hansenne,
and 
Carlos de Gois
for useful discussions and comments.
We acknowledge the use of IBM Quantum services for this work. The views expressed are those of the authors, and do not reflect the official policy or position of IBM or the IBM Quantum team.
This work has been supported by the Deutsche Forschungsgemeinschaft (DFG, German Research Foundation, Projects No. 447948357 and No. 440958198), the Sino-German Center for Research Promotion (Project No. M-0294), and the German Ministry of Education and Research (Project QuKuK, BMBF Grant No. 16KIS1618K). 
J.L.B. acknowledges support from the House of Young Talents of the University of Siegen. A.C. is supported by the EU-funded project \href{https://www.doi.org/10.3030/101070558}{FoQaCiA} Foundations of Quantum Computational Advantage and the \href{https://www.doi.org/10.13039/501100011033}{MCINN/AEI} (Project No.\ PID2020-113738GB-I00).


\appendix


\section{Unbiased estimators}\label{App_estimators}


In this appendix, we show that the estimators used in the main text are unbiased.


\subsection{Estimator in the infinite measurement limit}\label{App_estimator_inf}


First, we assume that the expectation values can be inferred directly, i.e.,
that we can repeat the measurement of the operator infinite times.
In this case, the estimator is given by Eq.~\eqref{Eq_estimator_inf}.
The expectation value has to be calculated with respect to the random 
variables $J_{l}$. 
With $\Es{\exs{B_{J_{l}}}}=\sum_{j=1}^{M}p(J_{l}=j)\exs{B_{j}}
=\frac{1}{M}\sum_{j=1}^{M}\exs{B_{j}}$, we obtain
\begin{equation}
\begin{split}
    \Es{\exs{\hat{\mathcal{B}}}_{\infty}}
    =&\frac{M}{L}\sum_{l=1}^{L}\Es{\exs{B_{J_{l}}}}
    =\frac{M}{L}\sum_{l=1}^{L}\frac{1}{M}\sum_{j=1}^{M}\exs{B_{j}}\\
    =&\sum_{j=1}^{M}\exs{B_{j}}
    =\exs{\mathcal{B}}.
\end{split}
\end{equation}


\subsection{Estimator for finite repetitions}\label{App_estimator_rep}


To calculate the expectation value of the estimator in Eq.~\eqref{Eq_estimator_rep},
we note that both the measurement outcomes $b_{j}$ and the index $J$ of the terms 
are random variables.
Hence, the expectation value of the estimator has to be taken over both
the measurement outcomes as well as the random picking, i.e., over $J$.
To evaluate the expectation value, we can thus make use of the law of iterated expectation. That is,
\begin{equation}
    \mathbb{E}[\ldots]=\mathbb{E}_{J}\left\{\mathbb{E}_{b_{J_{l}}}[\ldots|J_{l}=J]\right\}.
\end{equation}
This results in
\begin{equation}
\begin{split}
    \mathbb{E}[\exs{\hat{\mathcal{B}}}]
    &=\frac{M}{KL}\sum_{l=1}^{L}\sum_{k=1}^{K}\mathbb{E}_{J}\Big\{
        \underbrace{\mathbb{E}_{b_{J_{l}}}[b_{J_{l}}^{(k)}|J_{l}=J]}_{=\expval{B_{J}}}\Big\}\\
    &=\frac{M}{KL}\sum_{l=1}^{L}\sum_{k=1}^{K}\mathbb{E}_{J}[\expval{B_{J}}]\\
    &=\frac{M}{KL}\sum_{l=1}^{L}\sum_{k=1}^{K}\underbrace{\sum_{j=1}^{M}p(j)\expval{B_{j}}}_{
        =\sum_{j=1}^{M}\frac{1}{M}\expval{B_{j}}}\\
    &=\sum_{j=1}^{M}\expval{B_{j}}=\expval{\mathcal{B}}.
\end{split}
\end{equation}


\section{Hoeffding's inequality}\label{App_Hoeffding}


\subsection{Estimator in the infinite measurement limit}\label{App_Hoeffding_inf}


The estimator in Eq.~\eqref{Eq_estimator_inf} can be written as a sum of random
variables as follows:
\begin{equation}
    \exs{\hat{\mathcal{B}}}_{\infty}=\sum_{l=1}^{L}\underbrace{\frac{M}{L}
        \exs{B_{J_{l}}}}_{\eqqcolon X_{l}}.
\end{equation}
Since each term $B_{j}$ in the Bell operator is a tensor product of Pauli operators, 
$\exs{B_{j}}\in [-1,1]$ and thus $-\frac{M}{L}=a_{l}\leq X_{l}\leq b_{l}=\frac{M}{L}$.
Moreover, the bounded random variables $X_{l}$ are independent,
as they are obtained from different experimental runs.
We can thus use Hoeffding's inequality \cite{Hoeffding1963}, which states that
\begin{equation}
\begin{split}
    \Probs{\exs{\hat{\mathcal{B}}}_{\infty}-\exs{\mathcal{B}}\geq t}
    &\leq\exp(-\frac{2 t^{2}}{\sum_{l=1}^{L}(b_{l}-a_{l})^{2}})\\
    &=\exp(-\frac{2t^{2}}{\sum_{l=1}^{L}(\frac{2M}{L})^{2}})\\
    &=\exp(-\frac{t^{2}}{2M^{2}}L).
\end{split}
\end{equation}


\subsection{Estimator for finite repetitions}


As in the previous section, we can apply Hoeffding's inequality
to the estimator in Eq.~\eqref{Eq_estimator_rep}.
Also the estimator in Eq.~\eqref{Eq_estimator_rep} is a sum of independent 
random variables,
\begin{equation}
    \exs{\hat{\mathcal{B}}}=\sum_{l=1}^{L}\sum_{k=1}^{K}
    \underbrace{\frac{M}{KL}b_{J_{l}}^{(k)}}_{\eqqcolon Y_{kl}}.
\end{equation}
Since the outcomes $b_{J_{l}}^{(k)}$ are obtained from different experimental runs,
they are independent and thus are the random variables $Y_{kl}$.
In addition, the outcomes can only take the values 
$b_{J_{l}}^{(k)}\in\{-1, 1\}$.
Therefore, we have that the random variables $Y_{kl}$ are bounded as 
$-\frac{M}{KL}=a_{kl}\leq Y_{kl}\leq b_{kl}=\frac{M}{KL}$.
Finally, we get from Hoeffding's inequality,
\begin{equation}
\begin{split}
    \Probs{\exs{\hat{\mathcal{B}}}-\exs{\mathcal{B}}\geq t}
    &\leq\exp(-\frac{2t^{2}}{\sum_{l=1}^{L}\sum_{k=1}^{K}(b_{kl}-a_{kl})^{2}})\\
    &=\exp(-\frac{2t^{2}}{\sum_{l=1}^{L}\sum_{k=1}^{K}(\frac{2M}{KL})^{2}})\\
    &=\exp(-\frac{t^{2}}{2M^{2}}KL).
\end{split}
\end{equation}


\begin{figure}[t]
    \centering
    \includegraphics{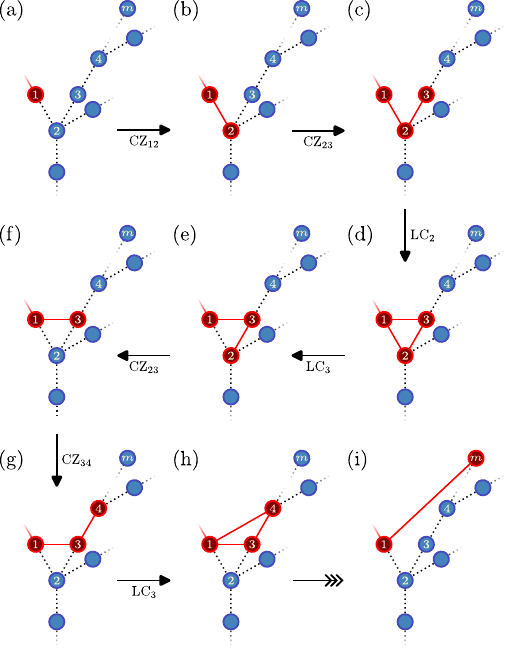}
    \caption{Scheme to decompose the CZ gate between arbitrary qubits $1$ and $m$ into
        a sequence of CZ gates between adjacent qubits and local complementations.
        The dotted lines denote the CZ gates that can be performed, whereas
        the red lines indicate the graph state.
        Qubit $1$ can already be coupled to different qubits, whereas the qubits $2,\ldots m$
        have to be uncoupled.
    \label{Fig_lemma_LC_prep}}
\end{figure}


\section{Preparation scheme for the LC state}\label{App_prep_LC_state}


We discuss a scheme to prepare a LC state with all qubits of an $n$-qubit quantum 
computer. This can be done by preparing all qubits in the $\ket{+}$ state and then 
applying CZ gates between some of them. A problem can be that the connectivity
of the quantum computer does not allow one to perform a specific gate between qubits
$i$ and $j$ directly. The following lemma shows that this is not a fundamental problem.

\begin{lemma}\label{Lemma_CZ_decomposition}
Consider a qubit array with a connected connectivity graph, where a CZ gate should
be applied to two qubits for graph state generation from the state $\ket{+}^{\otimes n}$. 
This can be achieved by a sequence of CZ gates between adjacent qubits 
(in the sense of the connectivity graph) and local complementations.
\end{lemma}
\begin{proof}
    We give an explicit construction that is visualized in
    Fig.~\ref{Fig_lemma_LC_prep}.
    The initial state is shown in Fig.~\ref{Fig_lemma_LC_prep}(a).
    We would like to perform a CZ gate between qubits $1$ and $m$, i.e., 
    $\text{CZ}_{1m}$.
    The interaction topology, however, does not allow a direct coupling.
    Rather, the qubits $1$ and $m$ are connected by the path of qubits $1,2,\ldots m$.
    Here, all the qubits  $1,2,\ldots m$ should be in the $\ket{+}$ state; 
    in particular, it is important that no CZ gate has been applied yet to the qubits $2,\ldots m$.
    In this case, we can apply the following scheme:
    \begin{enumerate}
        \item[(1)] Connect qubit $2$ by performing the $\text{CZ}_{12}$ gate
            to generate the first $\text{PAIR}(1,2)$ [Fig.~\ref{Fig_lemma_LC_prep}(b))].
        \item[(2)] While $k<m$ transform $\text{PAIR}(1,k)\rightarrow \text{PAIR}(1,k+1)$
            by the following:
        \begin{enumerate}
            \item Couple qubit $k$ and $k+1$ by $\text{CZ}_{k,k+1}$ 
                [Fig.~\ref{Fig_lemma_LC_prep}(c)].
            \item Couple qubit $1$ and $k+1$ by a local complementation on qubit $k$, i.e.,
                $\text{LC}_{k}$ [Fig.~\ref{Fig_lemma_LC_prep}(d)].
            \item Cancel the CZ gate between qubits $1$ and $k$ by performing $\text{LC}_{k+1}$
                [Fig.~\ref{Fig_lemma_LC_prep}(e)].
            \item Cancel the CZ gate between qubits $k$ and $k+1$ by 
                the controlled-Z $\text{CZ}_{k,k+1}$
                [Fig.~\ref{Fig_lemma_LC_prep}(f)].
        \end{enumerate}
    \end{enumerate}
    This allows one to decompose the $\text{CZ}_{1m}$ gate into a sequence
    with circuit depth $3(m-2)+1$.
    We note that step (2) is only necessary for $m>1$ and requires three steps 
    as the local complementations in (b) and (c) can be combined.
\end{proof}

Lemma \ref{Lemma_CZ_decomposition} can be used to construct a linear cluster state
on an arbitrary interaction topology.
\begin{observation}\label{Obs_LC_preparation}
    On a quantum computer of $n$ qubits, it is possible to prepare an
    $n$-qubit LC state with $\mathcal{O}(n)$ circuit depth.
\end{observation}
\begin{proof}
    The connectivity of a quantum computer is a connected graph $G$.
    It thus has a spanning tree, i.e., a tree graph that covers all vertices
    of $G$.
    In turn, a tree graph can be covered by a LC state by the following steps.
    At the start, all qubits are assumed to be prepared in the state $\ket{+}^{\otimes n}$.
    \begin{enumerate}
        \item[(1)] We start at a leaf and successively couple the adjacent qubits in the 
            direction of the root by CZ operations.
        \item[(2)] At a branch-off, check whether the other branch has already been covered.
            If all other branches have already been covered, we continue step (1) in the 
            direction of the root.
            Otherwise, Lemma \ref{Lemma_CZ_decomposition} can be used to couple the 
            last qubit to a leaf in the uncovered branch.
            From there, we can continue again with step (1).
    \end{enumerate}
    The scheme is shown for an exemplary two-qubit connectivity in 
    Fig.~\ref{Fig_LC_prep_tree}.
    To investigate the circuit depth, we note that step (1) and (2) are executed 
    alternately. 
    We thus count the number of steps for each run.
    $k_{i}$ denotes the number of steps for the $i$th execution of step (1), whereas
    $l_{i}$ stands for the steps required for the $i$th execution of step (2).
    In step (1), adjacent qubits are consecutively coupled by CZ gates.
    We thus have $\sum_{i}k_{i}<n$.
    In each step (2), a qubit at distance $m_{i}$ is coupled and, from 
    Lemma \ref{Lemma_CZ_decomposition} we know that $l_{i}=3(m_{i}-2)+1$.
    As each branch is only passed once, we have $\sum_{i}(m_{i}-2)\leq n$.
    Moreover, there are less than $n$ branch-offs, i.e., $\sum_{i}1\leq n$.
    Therefore, we obtain $\sum_{i}l_{i}=\sum_{i}3(m_{i}-2)+1\leq 4n$.
    The final circuit depth of the scheme is thus upper bounded by
    $\sum_{i}(k_{i}+l_{i})\leq n+4n=5n$.
\end{proof}


\begin{figure}[t]
    \centering
    \includegraphics{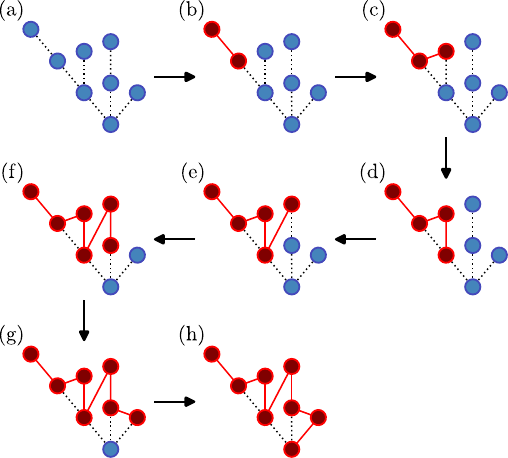}
    \caption{Illustration of the scheme in Observation~\ref{Obs_LC_preparation} for
        an exemplary two-qubit connectivity of eight qubits.
        The dotted lines denote the CZ gates that can be performed, whereas
        the red lines indicate the graph state.
    \label{Fig_LC_prep_tree}}
\end{figure}


\bibliography{literature}


\end{document}